\documentclass[lettersize,journal]{IEEEtran}
\usepackage{amsfonts,amsmath,amssymb,amsthm,bm}
\usepackage{booktabs,threeparttable}
\usepackage{graphicx,psfrag,epsf}
\usepackage{enumerate}
\usepackage{url}
\usepackage{algorithm}
\usepackage{algpseudocode}
\usepackage{helvet}
\usepackage{color} 
\usepackage{lineno}
\usepackage{chngcntr}
\usepackage{thmtools, thm-restate}

\theoremstyle{plain}

\newtheorem{theorem}{Theorem}

\usepackage{adjustbox}
\usepackage{dsfont}
\usepackage{cite}

\begin{document}
  \title{Discovering Communication Pattern Shifts in Large-Scale Labeled Networks using Encoder Embedding and Vertex Dynamics}
  \author{Cencheng Shen, Jonathan Larson, Ha Trinh, Xihan Qin, Youngser Park, Carey E. Priebe
  \thanks{
  \IEEEcompsocthanksitem Cencheng Shen is with the Department of Applied Economics and Statistics, University of Delaware. E-mail: shenc@udel.edu \protect
  \IEEEcompsocthanksitem Jonathan Larson and Ha Trinh
are with Microsoft Research at Redmond WA. E-mail: jolarso@microsoft.com, trinhha@microsoft.com \protect
   \IEEEcompsocthanksitem Xihan Qin is with the Department of Computer and Information Sciences, University of Delaware. E-mail: xihan@udel.edu \protect
  \IEEEcompsocthanksitem Carey E.Priebe and Youngser Park
are with the Department of Applied Mathematics and Statistics (AMS), the Center for Imaging Science (CIS), and the Mathematical Institute for Data Science (MINDS), Johns Hopkins University. E-mail: cep@jhu.edu, youngser@jhu.edu \protect
}
\thanks{This work was supported in part by the National Science Foundation HDR TRIPODS 1934979, 
the National Science Foundation DMS-2113099, and by funding from Microsoft Research.
The authors thank Ms. Ningyuan Huang, Mr. Tianyi Chen, and the anonymous reviewers for valuable comments and suggestions. The anonymized version of the email data on Microsoft that support this study will be retained indefinitely for scientific and academic purposes. The data are available from the authors upon reasonable request and with the permission of Microsoft.
}}
\markboth{IEEE Transactions on Network Science and Engineering}%
{Shell \MakeLowercase{\textit{et al.}}: Bare Demo of IEEEtran.cls for Computer Society Journals}

\maketitle

\begin{abstract} 
Analyzing large-scale time-series network data, such as social media and email communications, poses a significant challenge in understanding social dynamics, detecting anomalies, and predicting trends. In particular, the scalability of graph analysis is a critical hurdle impeding progress in large-scale downstream inference. To address this challenge, we introduce a temporal encoder embedding method. This approach leverages ground-truth or estimated vertex labels, enabling an efficient embedding of large-scale graph data and the processing of billions of edges within minutes. Furthermore, this embedding unveils a temporal dynamic statistic capable of detecting communication pattern shifts across all levels, ranging from individual vertices to vertex communities and the overall graph structure. We provide theoretical support to confirm its soundness under random graph models, and demonstrate its numerical advantages in capturing evolving communities and identifying outliers. Finally, we showcase the practical application of our approach by analyzing an anonymized time-series communication network from a large organization spanning 2019-2020, enabling us to assess the impact of Covid-19 on workplace communication patterns.
\end{abstract}



\section{Introduction}

\IEEEPARstart{G}{raph} data is a unique form of data structure that captures relationships between entities in various real-world settings, including social networks, communication networks, webpage hyperlinks, and biological systems \cite{GirvanNewman2002, newman2003structure, barabasi2004network, boccaletti2006complex, VarchneyEtAl2011, ugander2011anatomy}. A graph of $n$ vertices and $s$ edges can be represented by either an $n \times n$ adjacency matrix $\mathbf{A}$ or an $s \times 3$ weighted edgelist $\mathbf{E}$, with the latter being preferred for its storage efficiency. Graph data contains valuable information that can be used for various types of analysis, such as community detection \cite{KarrerNewman2011, ZhaoLevinaZhu2012, Hric2014}, link prediction \cite{liben2003link, leskovec2010predicting}, node classification \cite{perozzi2014deepwalk, kipf2017semi}, and outlier detection \cite{Ranshous2015, akoglu2015graph}, among others. 

Graph embedding is a highly versatile and widely used approach for analyzing graph data. It maps the nodes of a graph into a low-dimensional space while preserving the structural information of the graph. Unlike methods designed for specific tasks, graph embedding produces a vertex representation in Euclidean space that facilitates most downstream inference tasks. For instance, spectral embedding can reliably estimate the latent position of vertices \cite{RoheEtAl2011, SussmanEtAl2012}, perform community detection and vertex classification \cite{TangSussmanPriebe2013, SussmanTangPriebe2014, Tandon2021}, and handle multiple-graph inference and time-series graph embedding \cite{chen2020multiple, arroyo2021inference, Avanti2022, Patrick2022}. Popular and influential machine learning techniques, such as graph convolutional neural networks \cite{kipf2017semi, Wu2019ACS} and node2vec \cite{grover2016node2vec, node2vec2021}, are also examples of graph embedding methods.


In the modern digital era, graph data has become increasingly complex, with larger numbers of vertices and edges, and frequently appears as time-series data with evolving edge connectivity and weights over time. In the context of a large corporate environment with a workforce comprising hundreds of thousands of employees, their digital communication patterns can fluctuate on a monthly basis attributed to factors like organizational restructuring or global events such as the Covid-19 pandemic. On a grander scale, consider the ever-evolving landscape of social media platforms such as Meta and Twitter. As of 2023, Meta boasts an impressive 3 billion user base, while Twitter has over 350 million users. These examples highlights the dynamic nature of contemporary network data and the increasing need to analyze them effectively.

However, existing methods often require considerable computational resources to process such data and may not capture the evolving nature of dynamic networks. On the other hand, recent breakthrough in image and language analysis \cite{krizhevsky2012imagenet,radford2019languagegpt2} have demonstrated the crucial role of scalability in performance gain. The sheer amount of data often encodes sufficient information, and the ability to process vast amounts of data within a reasonable time frame can yield performance gain far exceeding that of a complicated and computation-intensive method limited to a small or subsampled dataset.

In this paper, we introduce a new approach called temporal encoder embedding for fast and scalable analysis of time-series networks. Utilizing either ground-truth or estimated vertex labels, our approach extends the concept of one-hot encoder embedding \cite{GEE1} to handle dynamic network data, delivering several notable advantages over existing methods. Firstly, it exhibits remarkable scalability, capable of processing billions of edges within minutes, surpassing the computational capacity of current methods. Secondly, our approach is computational stable without requiring additional complexities such as explicit dimension choice, random-walk schemes, or graph alignment, which are often needed by other techniques \cite{grover2016node2vec, arroyo2021inference, Patrick2022}. Thirdly, our method is theoretically sound, preserving the graph structure under random graph models with sufficiently large graph size. Finally, the resulting embedding provides temporal dynamic statistics, enabling the detection of communication pattern changes at multiple levels, spanning from individual outlier vertices to evolving community structures and encompassing entire network anomalies.

The simulation study employs the degree-corrected stochastic block model \cite{SnijdersNowicki1997, KarrerNewman2011,ZhaoLevinaZhu2012} to confirm the effectiveness of our approach in capturing stable networks, identifying outliers, and detecting pattern shifts. Additionally, we apply our method to a large-scale monthly communication network spanning 2019-2020, showcasing its efficiency and utility in real data analysis. The temporal encoder embedding enables rapid visualization and change-point detection at all levels of the network data, revealing both expected and unexpected behaviors during the Covid-19 pandemic. Our findings highlight the method's potential to analyze and comprehend complex time-series network data. The appendix contains proofs of theorems and detailed simulation information. The MATLAB code for both the method and simulations is available on GitHub\footnote{\url{https://github.com/cshen6/GraphEmd}}.

\section{Method}


\subsection{Temporal Encoder Embedding}

The method takes $T$ edgelists as input, all sharing a common set of $n$ vertices. These vertices should be accompanied by a ground-truth or estimated label vector $\mathbf{Y}$ of $K$ communities. The case of partial or no label vector is discussed in Section~\ref{sec:label}.

\begin{itemize}
\item \textbf{Input}: The edgelists $\{\mathbf{E}_{t} \in \mathbb{R}^{s_{t} \times 3}, t=1,\ldots,T\}$ and a label vector $\mathbf{Y} \in \{1,\ldots,K\}^{n}$.
\item \textbf{Step 1}: For each community $k=1,\ldots,K$, compute the number of observations per-community 
\begin{align*}
n_k = \sum_{i=1}^{n} I(\mathbf{Y}_i=k).
\end{align*}
\item \textbf{Step 2}: For the given label vector $\mathbf{Y}$, compute its one-hot encoding matrix $\mathbf{W} \in \mathbb{R}^{n \times K}$. Then, for each $k=1,\ldots,K$, normalize each column of $\mathbf{W}$ via
\begin{align*}
\mathbf{W}(\mathbf{Y}=k, k)=\mathbf{W}(\mathbf{Y}=k, k)/n_k.
\end{align*}
If $n_k=0$, the $k$th column is set to $0$ instead.
\item \textbf{Step 3}: At each time step $t=1,\ldots, T$, compute 
\begin{align}
\label{eq1}
\mathbf{Z}_{t}=\mathbf{A}_{t} \times \mathbf{W},
\end{align}
where $\mathbf{A}_{t}$ represents the adjacency matrix of $\mathbf{E}_{t}$.
\item \textbf{Step 4}: Denote each row of $\mathbf{Z}_{t}$ by $\mathbf{Z}_{t}(i, \cdot)$, and normalize every non-zero row by the Euclidean norm. Namely, for each $i$ and each $t$ where $\|\mathbf{Z}_{t}(i, \cdot)\|_2 >0$, compute 
\begin{align*}
\mathbf{\tilde{Z}}_{t}(i, \cdot)=\frac{\mathbf{Z}_{t}(i, \cdot)}{\|\mathbf{Z}_{t}(i, \cdot)\|_2}.
\end{align*}
\item \textbf{Output}: The temporal encoder embedding $\{\mathbf{\tilde{Z}}_{t} \in \mathbb{R}^{n \times K}, t=1,\ldots,T\}$.
\end{itemize}
Each vertex in the network has an embedding represented by $\mathbf{\tilde{Z}}_{t}(i, \cdot)$, where the $k$-th dimension $\mathbf{\tilde{Z}}_{t}(i, k)$ corresponds to the average connectivity of vertex $i$ to community $k$. The embedding dimension is fixed at $K$, which is the number of communities in the data. 

The normalization step creates a weighted-averaged representation of edge connectivity. Unlike other methods, aligning embedding across time is unnecessary as our approach estimates the average connectivity directly and is non-random, which means that the exact same graph will always yield the exact same embedding. 

The name temporal encoder embedding reflects the fact that the method is tailored for time-series graphs and involves the use of one-hot encoding in step 2. An alternative perspective to view this approach is that it performs a deterministic graph convolution on the one-hot labels. Recent research has shown that incorporating label structure into graph learning can enhance the learning performance \cite{huang2021combining, wang2022combining}.

\subsection{Temporal Dynamic Statistics}
Without loss of generality, assume we take time point 1 as the reference time point, the temporal dynamic statistics can be computed by taking the inner product of the vertex embedding.
\begin{itemize}
\item \textbf{Step 5}: The vertex dynamic statistic for vertex $i$ at time $t$ is calculated by
\begin{align*}
\mbox{Dynamic}_{1,t}(i)=1-<\mathbf{\tilde{Z}}_{t}(i, \cdot), \mathbf{\tilde{Z}}_{1}(i, \cdot)>,
\end{align*}
where $<\cdot,\cdot>$ denotes the standard inner product, the first subscript denotes the reference time $1$, and the second subscript denotes the current time $t$.

The community dynamic statistic for community $k$ is the average vertex dynamic statistic within the community:
\begin{align*}
\mbox{Dynamic}_{1,t}^{k}= 1- \sum_{i=1, \mathbf{Y}_i=k}^{n} \frac{<\mathbf{\tilde{Z}}_{t}(i, \cdot), \mathbf{\tilde{Z}}_{1}(i, \cdot)>}{n_{k}}.
\end{align*}
The graph dynamic statistic at time $t$ is the mean dynamic of all vertices:
\begin{align*}
\mbox{Dynamic}_{1,t}= 1- \sum_{i=1}^{n} \frac{<\mathbf{\tilde{Z}}_{t}(i, \cdot), \mathbf{\tilde{Z}}_{1}(i, \cdot)>}{n}.
\end{align*}
\end{itemize}
The temporal dynamic statistics lie in $[0,1]$ for non-negative weighted graphs, with higher values indicating greater deviation from the reference time point. A dynamic statistic of $0$ indicates that the vertex / community / graph connectivity at time $t$ is identical to the reference time point. Note that the dynamic statistics can be defined with respect to any other reference time point, such as the previous week to capture recent weekly pattern change. Alternatively, the maximum vertex dynamic within a period of time may also be used, i.e., $\max_{t=1,2,\ldots,T} {\mbox{Dynamic}_{1,t}(i)}$.

\subsection{On Label Vector}
\label{sec:label}
An important aspect of the algorithm requires the availability of a label vector. Depending on its availability, special consideration and extra processing may be needed before applying the default algorithm.

Scenario 1 pertains to situations where labels change over time, such as when a vertex initially belongs to group 1 but is later reassigned to group 2 in a communication network. It is certainly possible to utilize individual label vectors at each time step, i.e., use $\mathbf{W}_t$ in Equation~\ref{eq1} via an updated label vector at $t$. However, we recommend to use the label vector at the reference time consistently across all graphs, rather than using individual / different label vectors at each time. 

This is because the default method yields an embedding that only changes as the edge connectivity changes, whereas a multi-label approach yields an embedding that is influenced by changes in both labels and communication patterns. For example, if community $1$ is later partitioned into community $1$ and $2$, but all vertices in both groups still have the same connectivity as before, a multi-label approach would consider the graph to have changed significantly due to the label change and introduction of a new dimension, while the default approach will assert that there is no change at all. As shown in Section~\ref{sec:sim3} Figure~\ref{fig8}, the default method can detect pattern changes with evolving communities, even though it always uses the starting label vector. 

Scenario 2 pertains to situations where the label information is not available for all vertices. When partial labels are known, it suffices to set unknown labels as unused by assigning them a value of $0$, ensuring that the embedding utilizes only the known labels. This approach worked well in the supervised learning task \cite{GEE1}.

Another option, in case of no known labels, is to estimate the label vector at the reference time. One such approach is an iterative ensemble version of encoder embedding \cite{GEEClustering}: it initializes a random label vector, performs the embedding, and then applies k-means clustering to improve the labels. Namely, repeat step 1 - 4, followed by k-means on the embedding to compute new labels, and stop when labels no longer change. Note that, as discussed in scenario 1, this only needs to be done once at the reference time, so the computation impact is limited.

There exist other options to rapidly compute a label vector, such as Leiden, Louvain, constant Potts model, or label propagation \cite{Louvain2008,Leiden2019, traag2011narrow, raghavan2007near}. For example, the real data we considered in this paper used Leiden algorithm at the starting time to produce such a label vector. Once a label vector is obtained at the reference time point, the default method can be applied as usual. 



\section{Scalability}
The proposed method provides significant computational advantage over existing approaches. The time and storage complexity for the entire algorithm, including embedding, normalization, and dynamic statistic computation, is $O(nk T+\sum_{t=1}^{T}s_{t})$ with an overhead constant of about $2$. Furthermore, the method does not require parameter selection nor graph matching, which ensures computational stability. 

Steps 1, 2, 4, and 5 have a time complexity of $O(nkT)$ and are straightforward. Step 3, which involves matrix multiplication and seems computationally expensive, actually can be implemented in $O(\sum_{t=1}^{T}s_{t})$ using two simple edgelist operations: for any $j$th edge in $\mathbf{E}_{t}$, denote $u=\mathbf{E}_{t}(i,1)$, $v=\mathbf{E}_{t}(i,2)$, $w=\mathbf{E}_{t}(i,3)$, Equation~\ref{eq1} is equivalent to iterating through each edge and computing
\begin{align*}
\mathbf{Z}_{t}(u,\mathbf{Y}(v))=\mathbf{Z}_{t}(u,\mathbf{Y}(v))+\mathbf{W}(v,\mathbf{Y}(v))*w, \\
\mathbf{Z}_{t}(v,\mathbf{Y}(u))=\mathbf{Z}_{t}(v,\mathbf{Y}(u))+\mathbf{W}(u,\mathbf{Y}(u))*w. 
\end{align*}
Therefore, step 3 does not require matrix multiplication and can be parallelized across different time steps for large $T$. This allows the method to be optimized for time and storage complexity through parallelization or a streaming process, for which the time complexity can be further improved to $O(n(k+T) +\max_{t=1,\ldots,T}s_{t})$. Furthermore, if the graph lacks any labels, and we employ the iterative ensemble approach \cite{GEEClustering} to estimate the label, the time complexity remains the same with a higher constant for $nK$.

Figure~\ref{fig3} provides a comparison among temporal encoder embedding, the unfolded spectral embedding (USE) \cite{Patrick2022}, and the graph convolutional neural network (GCN) \cite{kipf2017semi}. For the spectral embedding, we utilized a fast sparse implementation and computed only the top $30$ singular values and vectors \cite{Zhang2018SpectralNE}. For GCN, we used a fast, sparse, and GPU-based implementation from the official MATLAB documentation\footnote{\url{https://www.mathworks.com/help/deeplearning/ug/node-classification-using-graph-convolutional-network.html}}, limited to $100$ epochs. For the encoder embedding, we considered both the default implementation with known labels and the scenario without labels, where the iterative ensemble approach is first applied to estimate all labels at the reference time point.  

From the figure, it is evident that the encoder embedding with known labels is the fastest approach, typically being a few times faster than the version without labels due to the extra step required to estimate the labels. In both cases, the encoder embedding is significantly faster in magnitude than the spectral embedding and GCN. For example, at the last x-axis point on the left panel, encoder embedding (without labels) is $20$ times faster than USE and GCN. On the last x-axis point on the right panel, encoder embedding (without labels) is $20$ times faster than GCN and $140$ times faster than USE. These running time results were obtained on a Windows 10 machine with a 12-core Intel i7-6850 CPU, 64GB memory, NVIDIA 1080Ti GPU, and MATLAB 2022a. 

\begin{figure}[ht]
    \centering
	\includegraphics[width=0.5\textwidth,trim={0cm 0cm 0cm 0cm},clip]{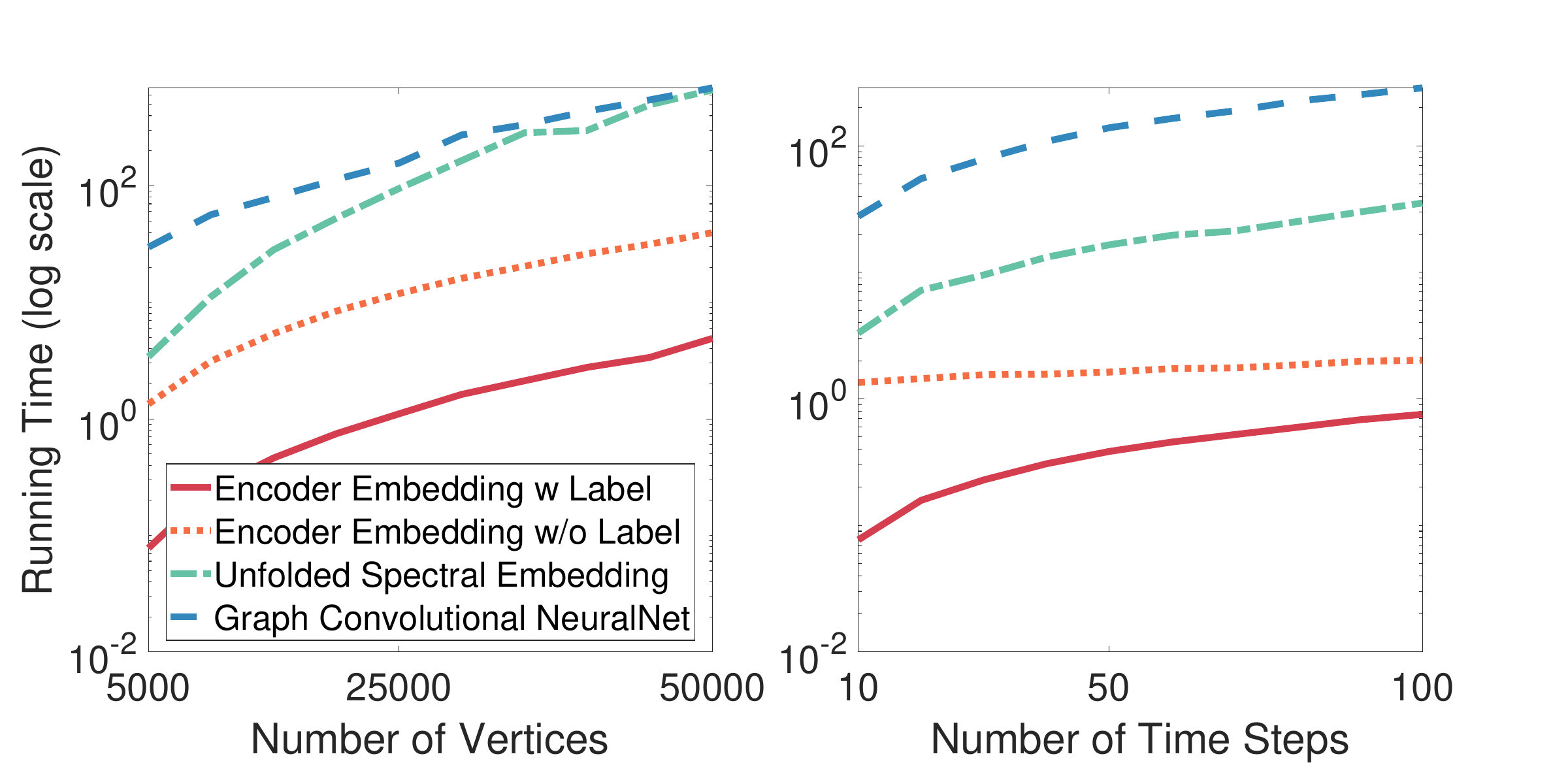}
	\caption{This figure presents a comparison among the encoder embedding with all labels known, the encoder embedding without labels, the unfolded spectral embedding, and graph convolution neural network. The average running time, including both the embedding and computation of dynamic statistics, was computed over 10 replicates. The experiments were conducted with a fixed $K=20$. The left panel varies $n$ from $5000$ to $50000$ with a fixed $t=10$, while the right panel varies $t$ from $10$ to $100$ with a fixed $n=5000$. }
	\label{fig3}
\end{figure}

Note that the results presented in Figure~\ref{fig3} are limited to graphs with up to $n=50,000$ vertices and $5$ million edges, due to the limitations of benchmark methods. 
In fact, the default method with labels, when tested on our local PC, can process a graph with $10$ million vertices and $1$ billion total edges within $5$ minutes. By comparison, the fastest spectral embedding method to date \cite{Zhang2018SpectralNE}, which utilizes sparse structure and has almost linear complexity, takes over $3$ hours to process a graph with $10$ million vertices and only $50$ million edges. Furthermore, Node2vec requires approximately $4$ hours to process a single graph with only $1$ million vertices and $10$ million edges \cite{grover2016node2vec}.

Assuming adequate storage capacity and processing power, we anticipate that the proposed method can embed graphs with $1$ billion vertices and $100$ billion edges within $10$ hours. This level of scalability allows for processing even the largest social networks on a daily or weekly basis. 

\section{Supporting Theory}
Here we establish the mathematical foundation behind the algorithm, under a conditional independent random graph model for the edges. Specifically, we assume that each edge $A_{t}(i,j) | i,j$ is independently generated by a certain distribution of finite second moments. This assumption encompasses many popular random graph models, including the stochastic block model \cite{SnijdersNowicki1997, KarrerNewman2011} and the degree-corrected variant \cite{ZhaoLevinaZhu2012}, as well as the random dot product models \cite{HollandEtAl1983, YoungScheinerman2007}. All proofs are in the appendix. 

\begin{theorem}
Assuming the conditional independent random graph model, the temporal encoder embedding converges to a conditional expectation. Specifically, for a vertex $i$ belonging to community $y$, we have that
\begin{align*}
Z_{t}(i,k) \stackrel{n \rightarrow \infty}{\rightarrow} \frac{a_{t}(i,k)}{\|a_{t}(i,:)\|_{2}},
\end{align*}
where $a_{t}(i,:) \in \mathbb{R}^{K}$ satisfies
\begin{align*}
a_{t}(i,k) = E(\mathbf{A}_{t}(i,j) | \mathbf{Y}_i=y, \mathbf{Y}_j=k).
\end{align*}
\end{theorem}
The conditional expectation, which is equivalent to the average connectivity of vertex $i$ to community $k$, can be used to capture changes in communication patterns over time. For instance, under the stochastic block model, $a(i,k)$ corresponds to the block probability vector for community $k$. Moreover, it has been shown that the encoder embedding is equivalent to the more computationally expensive spectral embedding up to transformation \cite{GEE1}. 

Now, a slight change in connectivity to any community is reflected in this conditional expectation and results in a different embedding position from a time-series perspective. 
\begin{theorem}
Assuming non-zero conditional expectations, the dynamic statistic of vertex $i$ at time $t$ converges to $0$ asymptotically if and only if $\theta(a_{t}(i,:), a_{1}(i,:))=0$, and it converges to $1$ asymptotically if and only if $\theta(a_{t}(i,:), a_{1}(i,:))= \pi/2$, where $\theta(\cdot,\cdot)$ denotes the angle between two vectors.
\end{theorem}
The theorem provides a geometric interpretation for the dynamic statistics. A vertex dynamic of $0$ means that the vertex maintains the same correspondence with the communities up to multiplication, for example, $a_t(i,:)=(2,2,4)$ and $a_1(i,:)=(1,1,2)$. Conversely, a maximum vertex dynamic suggests that the vertex has either stopped being active or moved to a space orthogonal to its position at time $1$, for example, $a_t(i,:)=(1,1,0)$ and $a_1(i,:)=(0,0,2)$. 
A vertex dynamic of $0.5$ implies a 30-degree angle difference in communication pattern. 
The same interpretations hold for the community and graph dynamic statistics. 

\section{Simulations}
In the first simulation, our objective is to demonstrate the stability of the resulting embedding in the context of a stable dynamic network with small noise. In the second simulation, we evaluate the method's capability to identify vertex outliers. In the third simulation, we verify its ability to capture community pattern changes while remaining robust against community reassignment. For graph generation, we consistently utilized the degree-corrected stochastic block model, which approximates real sparse graphs more effectively than alternative models such as the standard stochastic block model or random dot product model. More detailed information on the network generation process and model parameters can be found in the appendix.

\subsection{Stable Network}
This simulated network data was generated using a weighted version of the degree-corrected stochastic block model, incorporating varying edge weights over time. The network comprises $30,000$ vertices with positive edge weights within the range $[1,...100]$. Over a span of $96$ time steps, random noise gradually affects the edge weights, causing them to become more distinct as $t$ increases. The total number of edges across all time steps approximates $420$ million. The embedding process was completed in approximately $13$ seconds, with an additional $0.6$ seconds required for dynamic statistics computation.

The embedding is observed to be remarkably stable, as demonstrated in Figure~\ref{fig1} for the first three dimensions and communities. The same stability is observed in the vertex dynamic statistics, as illustrated in Figure~\ref{fig2}. Furthermore, the dynamic statistics effectively capture the incremental noise. At time step 2, the graph dynamics are below $0.001$, gradually increasing to below $0.01$ by time step 12, and eventually reaching $0.025$ at time step 96. The vertex dynamics follow a similar increasing trajectory, as evident in Figure~\ref{fig2}.  These results clearly demonstrate that our proposed approach successfully captures both the stable nature of the network and the effects of small incremental noise.

\begin{figure*}[ht]
    \centering
	\includegraphics[width=1.0\textwidth,trim={0cm 0cm 0cm 0cm},clip]{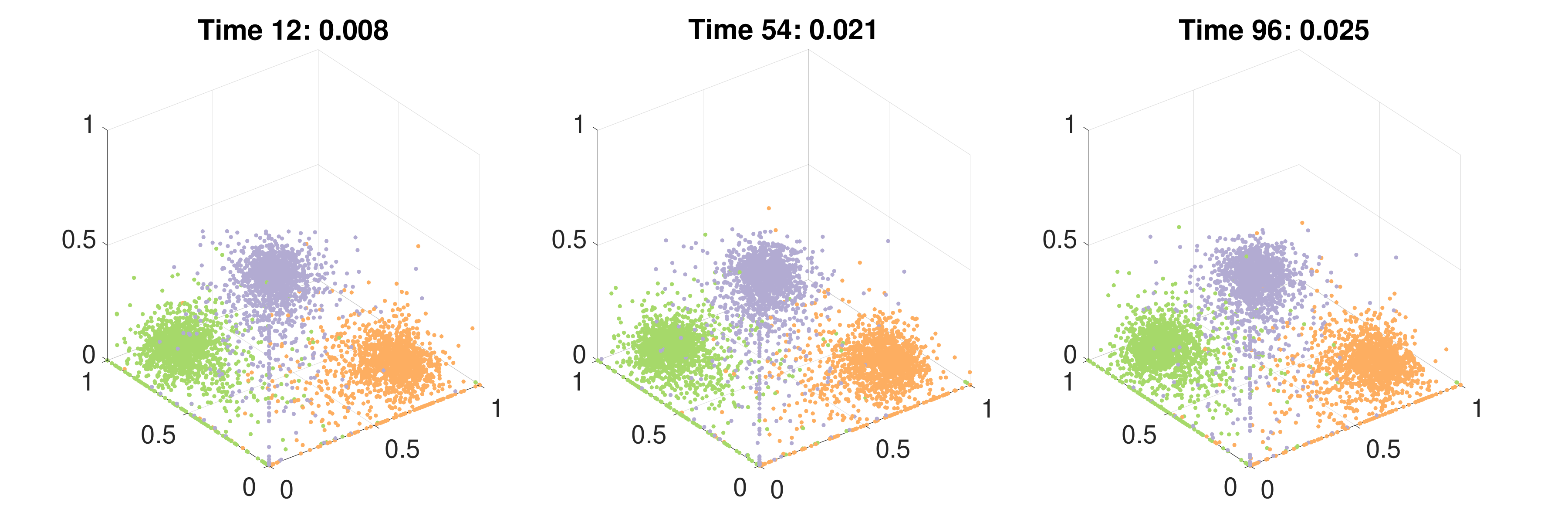}
	\caption{3D Visualization of the first 3 communities' vertices at three different times for the simulated graph. The graph dynamic at each time is shown on top. }
	\label{fig1}
\end{figure*}

\begin{figure*}[ht]
    \centering
	\includegraphics[width=1.0\textwidth,trim={0cm 0cm 0cm 0cm},clip]{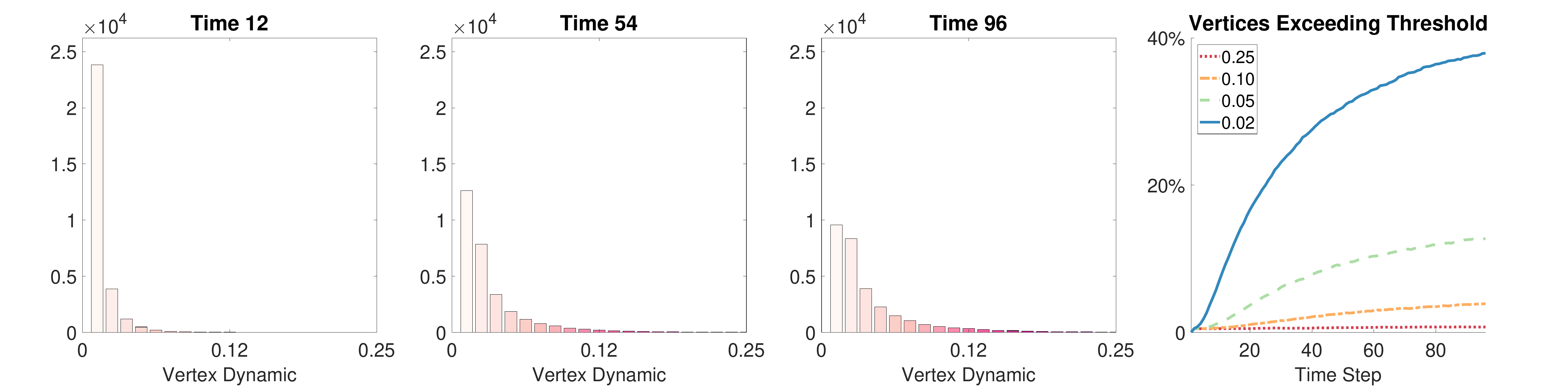}
	\caption{Visualization of the vertex dynamic statistics as time progresses. For the first 3 panels, the y-axis represents the number of vertices, while the x-axis represents the extent to which the vertices have shifted. As time increases, more vertices start to shift away from their starting positions due to noise. The last panel shows the percentage of vertices exceeding the vertex dynamic threshold at the last time step. }
	\label{fig2}
\end{figure*}

\subsection{Outlier Vertices}
In this simulation, we utilized the same model setting at a smaller scale, using $n=1000$ vertices and $T=10$ time steps. Furthermore, we introduced $10$ extreme outlier vertices at $t=10$, each having one or two incident edge weights randomly assigned within the range of $[500,1000]$. This outlier communication scenario created a challenge in detecting the outliers since they only occurred at the final time point and involved a limited number of edges, while the remaining edge weights ranged from $1$ to $100$.

For comparison purpose, we also consider the unfolded spectral embedding (USE), which has cross-sectional and longitudinal stability guarantee for time-series networks \cite{Patrick2022}. The data is embedded into $d=10$ dimensions via USE, followed by computing the per-vertex Euclidean distance between time $10$ and time $1$ as the outlier measure. The top row of Figure~\ref{fig7} shows the histogram of outlier measure. The vertex dynamic statistic by encoder embedding reveals that most vertices underwent minimal shifts, while USE suggests that a larger number of vertices experienced significant shifts. The bottom panel of the figure shows the magnitude and percentile ranking of the $10$ outliers. It is evident that temporal encoder embedding outperformed USE in identifying outliers. Specifically, the outliers identified using temporal encoder embedding achieved a vertex dynamic ranking of $[1,2,3,4,5,7,8,9,14,36]$, while the corresponding ranking using USE was $[42,69,89,114,147,168,235,318,411,475]$, indicating that these outliers were obscured among the other vertices. These findings remained consistent when considering different values of $d$ in the USE method.


\begin{figure}[ht]
	\includegraphics[width=0.5\textwidth,trim={0cm 0cm 0cm 0cm},clip]{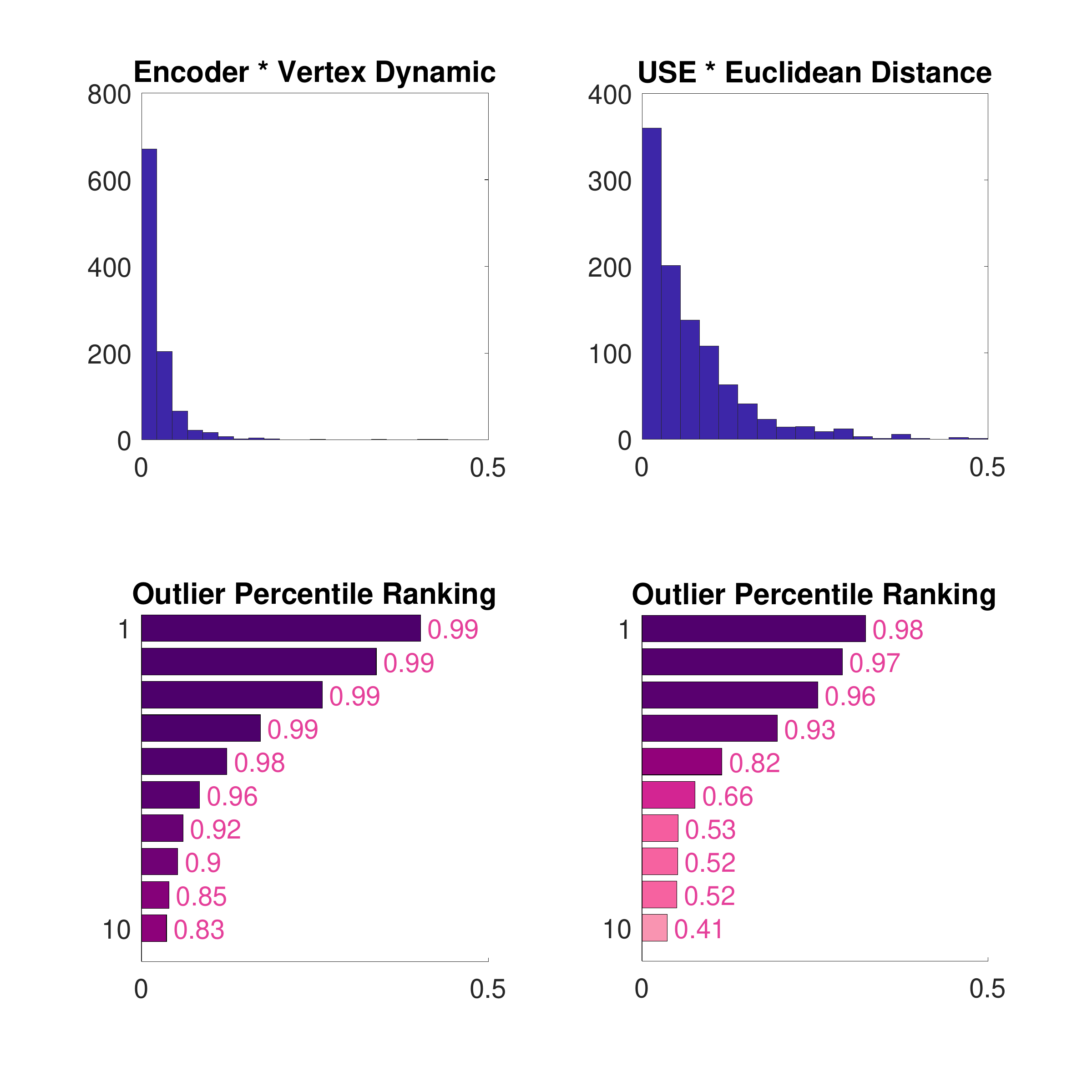}
	\caption{This figure compares temporal encoder embedding and unfolded spectral embedding in detecting $10$ extreme outliers.}
	\label{fig7}
\end{figure}

\subsection{Pattern Shift with Community Change}
\label{sec:sim3}
In this simulation, we explore various pattern shift scenarios as time progresses. Beginning with $n=30000$ vertices and $K=3$ communities, the first graph at $t=1$ is generated using a degree-corrected stochastic block model, where within-community vertices are more likely to be connected than between-community vertices. As shown in the top left panel of Figure~\ref{fig8}, this initial communication pattern is effectively captured by the encoder embedding. In the figure, different colors represent different communities, and these communities are distinctly separated, indicating their distinct communication patterns. It is worth noting that such a pattern is common in many organizations, where team members within a group tend to communicate more frequently than with members from other groups.

At $t=2$, the second graph is generated to model typical shifting pattern. Everything else remains the same as at $t=1$, vertices from community $3$ are now equally likely to connect with every other vertex. This essentially brings community $3$ much closer to the other two communities. This change is evident in the top right panel of Figure~\ref{fig8}, where community $3$ (represented by the blue dots on top) appears closer to the other two communities in the resulting embedding. Note that such a change is common in a corporate setting, for instance, due to a shift in work scope where individuals in community $3$ now need to communicate equally often with everyone in the company.

At $t=3$, the third graph is generated to model a shifting pattern involving a community split. Keeping everything else the same as at $t=2$, each vertex in community $3$ is randomly reassigned to a new community $4$ with a $50\%$ probability. Community $4$ retained the same communication pattern as community $3$ had at $t=1$. The bottom left panel of Figure~\ref{fig8} clearly captures this change, as communities $3$ and $4$ now have distinct embedding. Moreover, when comparing the embedding at $t=3$ to that at $t=1$, it is noticeable that this new community $4$ at $t=3$ (the brown dots on top) occupies a similar Euclidean position as community $3$ at $t=1$. Such splits are typical in organizational restructuring within a corporate setting.

At $t=4$, the fourth graph is generated to model a shifting pattern involving a community merge. In this scenario, community $3$ is merged with community $1$ and set to the same connectivity as community $1$. The bottom right panel of Figure~\ref{fig8} illustrates that the embedding of communities $3$ and $1$ essentially merge together. Such mergers are also typical in organizational restructuring within a corporate setting.

Note that all four graphs and the initial label vector were used as input to generate the embedding. The updated label vector resulting from the community label split at $t=3$ was not used in the embedding process. Instead, it was employed solely for the purpose of figure visualization, allowing us to track whether the pattern change matches the ground truth. The visualization demonstrates that using the same label in the temporal encoder embedding successfully captures the pattern shifts, even in the presence of significant community changes. 

Taking $t=1$ as the reference time point, the community dynamic statistics from $t=2$ to $t=4$ are as follows:
\begin{itemize}
\item Community 1: $0.03, 0.01, 0.09$; 
\item Community 2: $0.03, 0.01, 0.01$; 
\item Community 3: $0.31, 0.17, 0.22$.
\end{itemize}
These statistics correspond to the starting label vector, so there is no community $4$. We interpret these statistics as follows: 
\begin{itemize}
\item Community $2$ experienced minimal changes throughout, with a slightly larger statistic at $t=2$ due to a slight increase in communication with community $3$.
\item The same pattern holds for community $1$ at $t=2$ and $t=3$, but the statistic becomes larger at $t=4$ due to its merger with new vertices.
\item For vertices in community $3$, they underwent significant changes at every time point, with the most substantial change occurring at $t=2$ when every vertex in community $3$ shifted their connectivity significantly. At $t=3$ and $t=4$, approximately half of the vertices shifted back to their original pattern at $t=1$, resulting in smaller but still significant dynamic statistics.
\end{itemize}
Overall, the temporal embedding successfully captures all intended pattern changes throughout this simulation, and the temporal dynamic statistics serve their purpose in summarizing the magnitude of the change. 

\begin{figure}[ht]
	\includegraphics[width=0.5\textwidth,trim={0cm 0cm 0cm 0cm},clip]{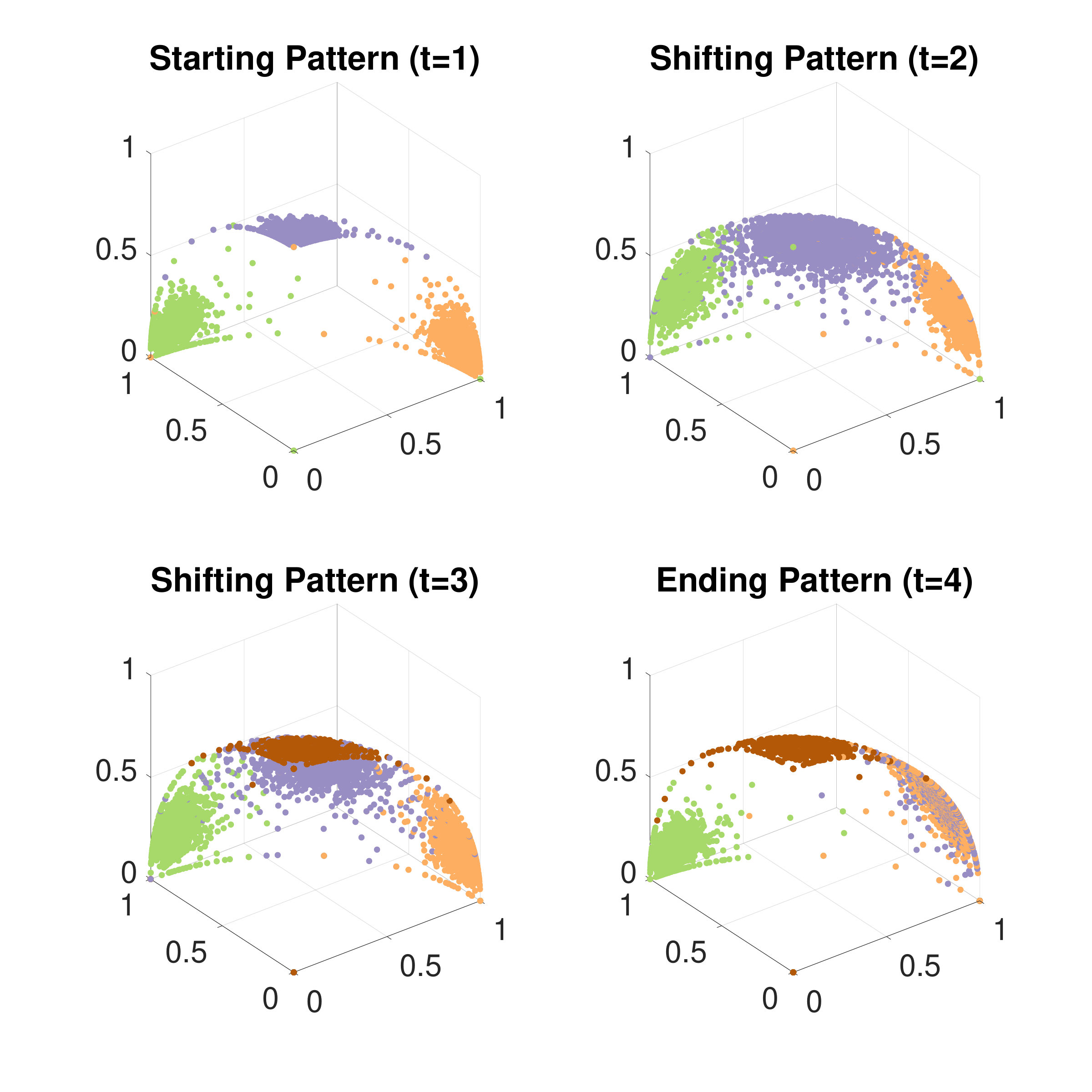}
	\caption{This figure visualizes how the encoder embedding successfully detects the changing communication pattern despite community label changes.}
	\label{fig8}
\end{figure}

\section{Communication Network during Covid-19 Pandemic}

The dataset used in this study was obtained from Microsoft, consisting of $116$ thousand anonymized and aggregated entities. The dataset spans 24 months, from January 2019 to December 2020, and contains over $80$ million weighted edges among $39$ organizational groups. Namely, $n=116,508$, $T=24$, and $K=39$. The edge weight represents volume of email communication between the connected vertices. Our objective is to analyze this data and evaluate the effects of the Covid-19 pandemic on this communication network. Please note that the $39$ vertex communities were estimated by applying the Leiden algorithm \cite{Leiden2019} (we used the Python implementation via the graph statistic package\footnote{\url{https://github.com/microsoft/graspologic}}) to the graph at the starting month, which took approximately $30$ seconds. 

The Covid-19 pandemic has had a significant impact on many corporations and how people work. It is a global event that has caused an increase in reliance on digital communication and remote work. It is important to gain an understanding of how the pandemic, work-from-home policies, and the rise of remote work have affected work patterns \cite{kleinbaum2013discretion, jacobs2021large, cep2022}. To that end, analyzing intra-organization communication networks, such as email and chat correspondence, can provide valuable insights into how pandemic-related changes have influenced work behavior, which is of interest for organizational restructuring purposes.

By employing a suitable graph embedding, we can obtain a Euclidean representation for each entity in the communication network and track pattern changes over time via a proper distance measure. This facilitates the identification of vertices and communities that have experienced significant changes as a result of the pandemic, as well as those that have remained unaffected by it.

\subsection{Embedding and Visualization}

Temporal encoder embedding was applied to this dataset and completed in $10$ seconds on a standard Windows 10 PC. The embedding is of size $116508 \times 39 \times 24$. To visualize the embedding, we further applied UMAP \cite{UMAP} to the graph embedding, which a fast and effective projection tool for 2D visualization. Figure~\ref{fig4} illustrates the visualization of the embedded network for 2019 June, 2019 Dec, and 2020 June. While the time-series graph retains a common structure across time, it is apparent that numerous groups and vertices undergo significant changes.

\begin{figure*}[ht]
\centering
\includegraphics[width=.32\linewidth]{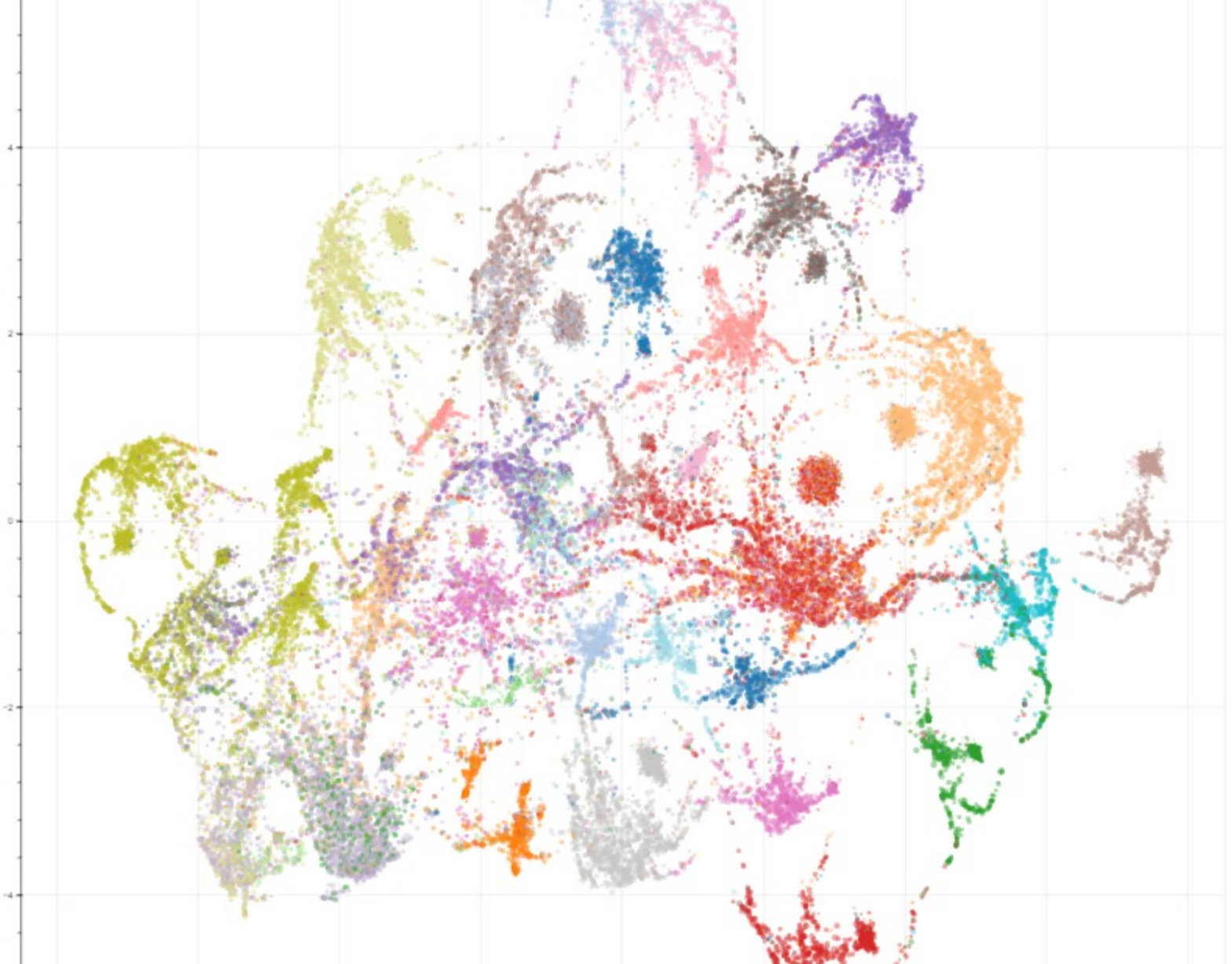}
\includegraphics[width=.32\linewidth]{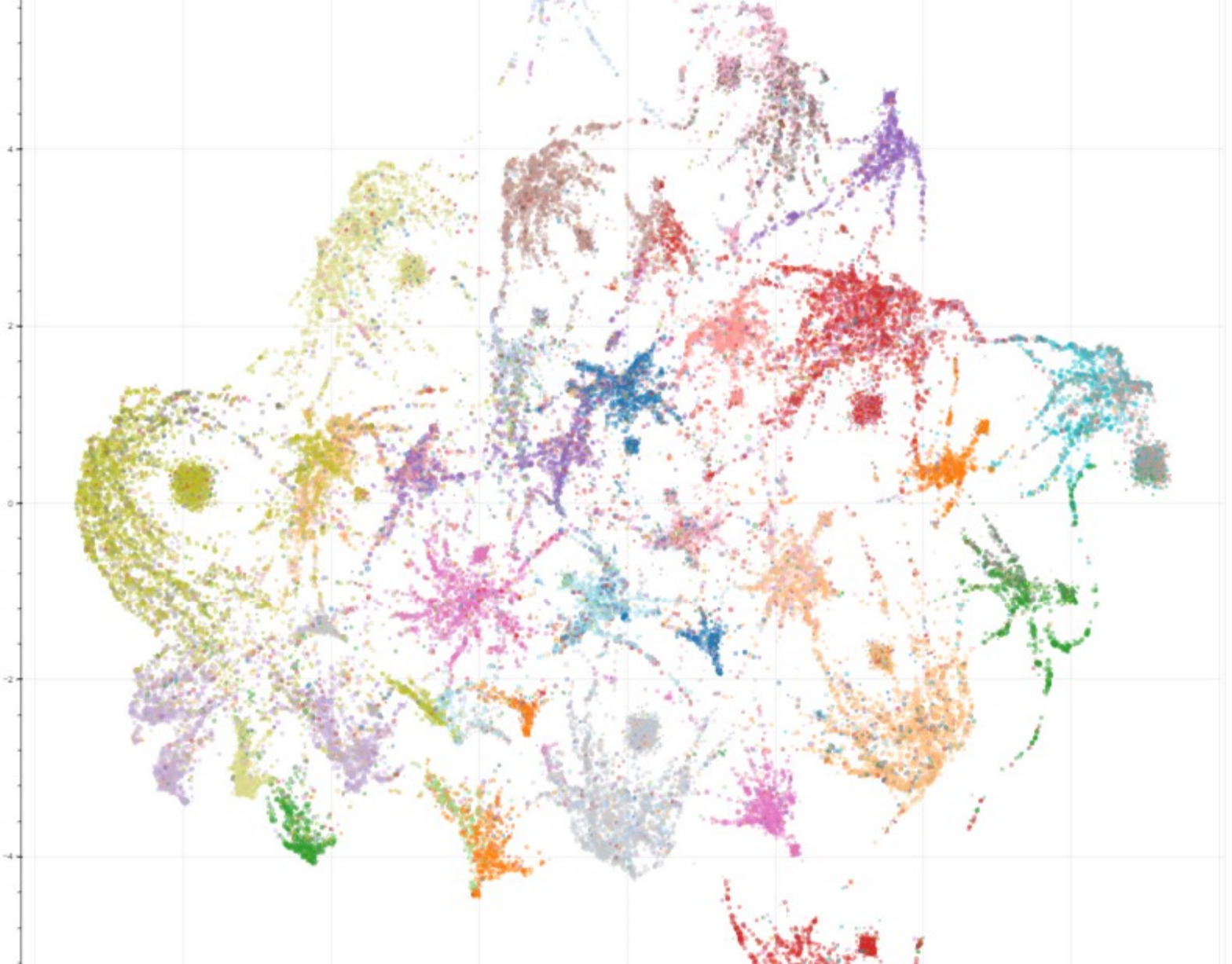}
\includegraphics[width=.32\linewidth]{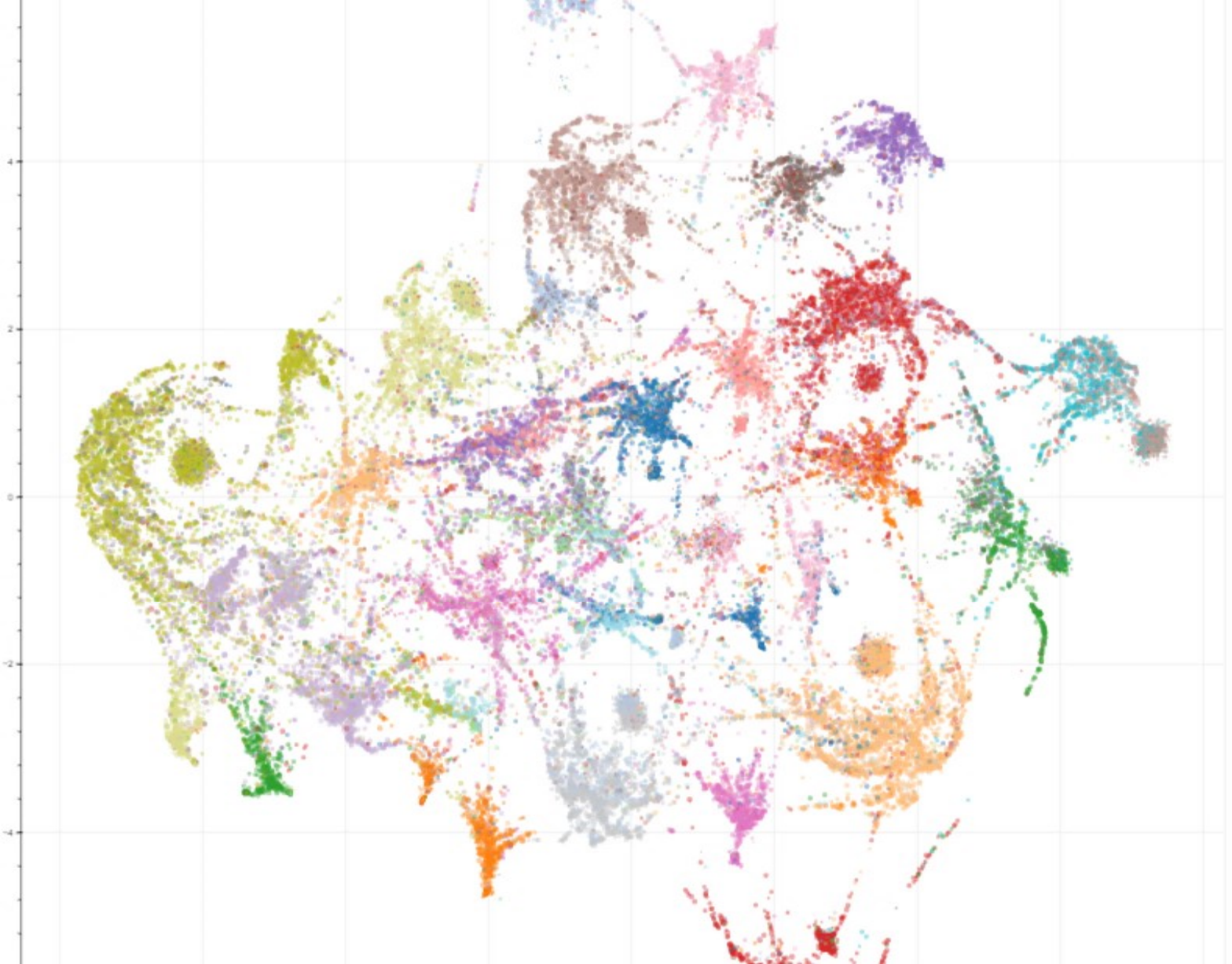}
\caption{The whole communication network visualized by temporal encoder embedding and UMAP. The color scheme differentiates the various groups within the organization. The proximity of individual nodes in the 2D projection indicates a higher frequency of communication between them, especially those belonging to the same group. The graph dynamic statistic is displayed on top.}
\label{fig4}
\end{figure*}

\subsection{Change Detection via Temporal Dynamics}
The temporal encoder embedding enables three types of temporal dynamic statistics to detect communication pattern changes: vertex dynamic statistics for each of the $116,508$ vertices over 24 months, community dynamic for each of the $39$ communities, and graph dynamic for the entire network. These dynamic statistics lie in $[0,1]$, with lower values indicating minimal change in edge connectivity and higher values indicating a greater deviation from the baseline time point. It is worth noting that the dynamic statistics necessitate a baseline time point, which was set as January 2019, but it can be modified to any other month, such as January 2020, to focus solely on the pandemic's impact.

The computation of temporal dynamic statistics took approximately one second after generating the embedding. Figure~\ref{fig5} displays the temporal dynamics for selected vertices and communities, representative of changes in communication patterns. The left panel shows that vertex $29$ and $66$ experienced significant changes around March 2020 due to the pandemic, while vertex $8$ underwent similar but smaller changes. In contrast, vertex $7$ had more gradual changes in communication patterns, while vertices $6$ and $41$ exhibited communication patterns that were relatively unaffected by the pandemic.

The right panel of Figure~\ref{fig5} displays changes in the community and graph dynamic statistics over time. The graph dynamic statistics demonstrate that pattern changes were occurring even before the pandemic, with values of $0.2$ by December 2019, $0.31$ by June 2020, and $0.33$ by December 2020. The pandemic caused an accelerated increase in early 2020, but the change stabilized after June 2020. We speculate that remote work was already influencing the communication structure before the pandemic, and Covid-19 simply accelerated and completed the trend, leading to a relatively stable communication network after the pandemic. In addition, the community dynamic identifies communities that experienced significant changes and those that remained relatively unchanged. Communities 8 and 32 shifted their communication patterns, while communities 3 and 39 were minimally impacted by the pandemic. Note that the community and graph dynamic statistics are more smooth than vertex dynamic statistics due to their aggregated nature. 

Finally, different thresholds can be applied to the temporal dynamic statistics to identify outlier and inlier vertices. For example, using a threshold of $0.5$ for the vertex dynamic identifies $11\%$ outlier vertices in June 2019, $17\%$ in December 2019, $29\%$ in June 2020, and $31\%$ in December 2020. This again suggests a trend of change before the pandemic that was significantly accelerated in March and April of 2020 before stabilizing. A histogram of the temporal vertex dynamic in Figure~\ref{fig6} confirms this trend. Conversely, a threshold of $0.1$ identifies inliers, with $73\%$ in June 2019, $63\%$ in December 2019, $48\%$ in June 2020, and $44\%$ in December 2020. These results suggest that almost half of the vertices maintained their pre-pandemic communication patterns, either because remote work was already part of their work-style or because their work could not be performed remotely.

\begin{figure}[ht]
    \centering
	\includegraphics[width=0.5\textwidth,trim={0cm 0cm 0cm 0cm},clip]{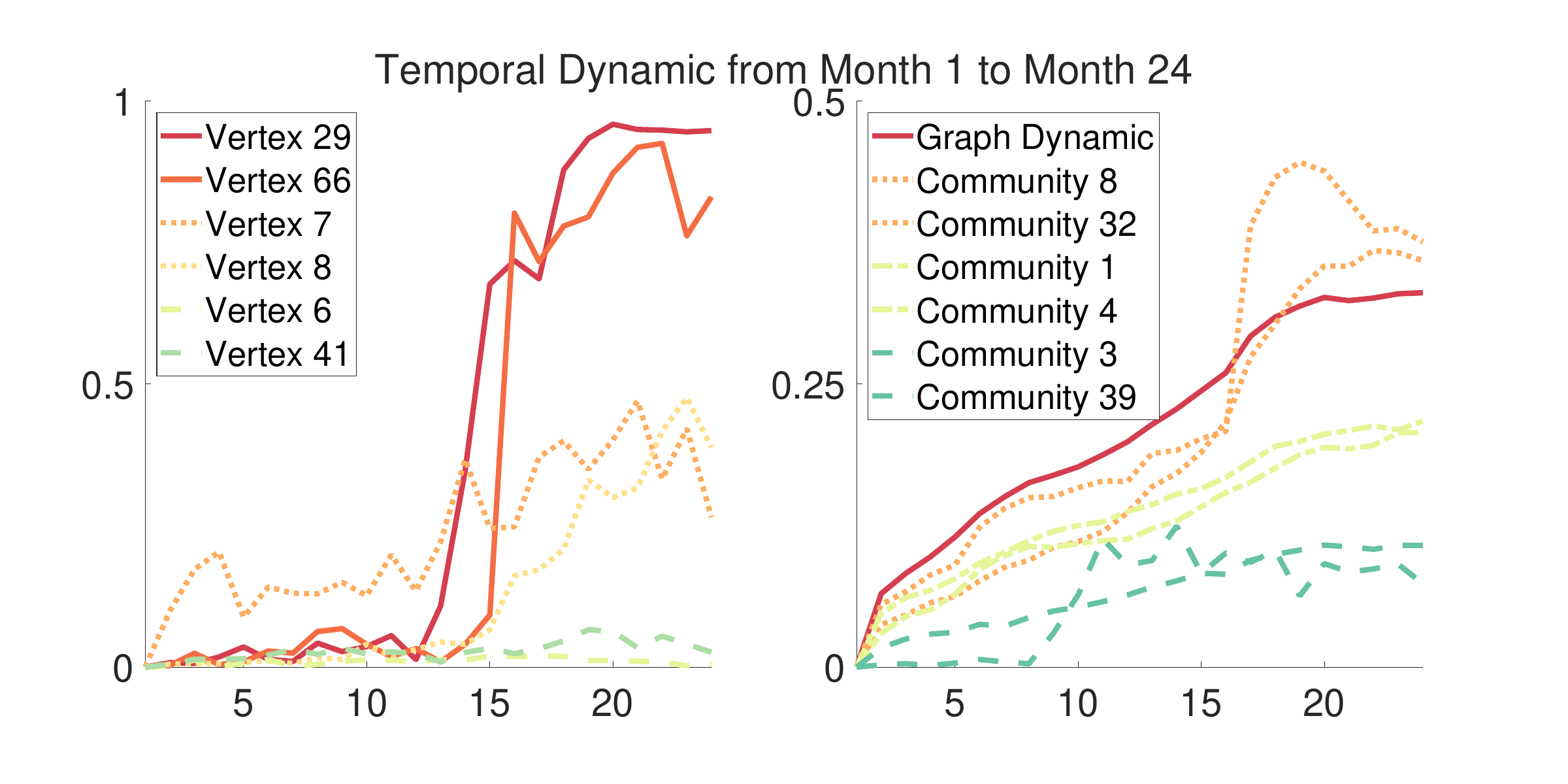}
	\caption{Visualization of the temporal dynamic statistics. The left panel shows the vertex dynamics for several selected vertices. The right panel shows the graph and the community dynamic statistics for chosen communities. }
	\label{fig5}
\end{figure}

\begin{figure*}[ht]
    \centering
	\includegraphics[width=1.0\textwidth,trim={0cm 0cm 0cm 0cm},clip]{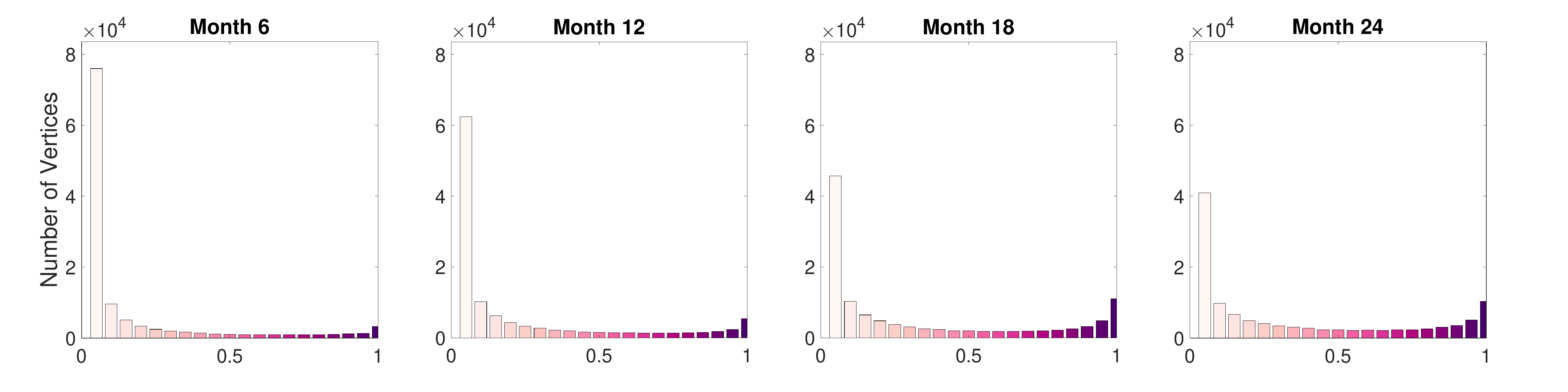}
	\caption{Histogram showing the distribution of vertex dynamic statistics at different time points. }
	\label{fig6}
\end{figure*}

\section{Conclusion}
The paper presents a fast and efficient graph embedding method for large-scale dynamic networks. The proposed method has linear computational and storage complexity with respect to the number of vertices and edges, making it highly scalable and capable of embedding even billions of vertices on a daily basis. The synthetic study confirms the method's numerical advantages, and its asymptotic behavior is characterized by conditional expectation under a random graph model. It was successfully applied to a large communication network during the Covid-19 pandemic, revealing significant communication pattern shifts and identifying individual vertices and communities most or least impacted by the pandemic. The proposed method not only allows for direct application to large-scale graph data, but also enables further method development and theoretical understanding within this scalable framework.

There are several open areas for further investigation, and one of these areas pertains to the reference time and its impact on the method. The reference time influences the approach in two key aspects. Firstly, the dynamic statistics are defined with respect to the reference time. Secondly, in cases where ground-truth labels are absent, the label vector is estimated using the graph data at the reference time. The first aspect is straightforward, as dynamic statistics and pattern shifts are inherently relative. Choosing a different reference time merely means that any pattern shifts are assessed in relation to the graph structure at that specific reference time. For instance, community $1$ might have shifted significantly from time $1$ to time $100$, resulting in substantial dynamic statistics. However, if time $50$ is chosen as the reference time instead of time $1$, the dynamic statistics are computed with respect to time $50$, which may yield smaller statistics.

The second aspect is particularly profound and intriguing. Choosing a different reference point can result in entirely different labels, which, in turn, can significantly impact the final embedding. While simulations have demonstrated that a significant amount of pattern shifts are discernible as long as the same label is used consistently across each time step, it is certainly possible that using a coarse label vector may obscure some pattern changes. For example, when working with data that has a ground-truth of $10$ classes, estimating only $9$ classes for labeling can certainly hide some information in the embedding. Conversely, if a finer label vector is estimated, it may have the potential to reveal previously hidden structural shifts. Therefore, the impact of the estimated label on the method is an interesting area for further exploration.



\bibliographystyle{ieeetr}
\bibliography{shen,general}

\begin{IEEEbiography}
[{\includegraphics[width=1in,height=1.25in,clip,keepaspectratio]{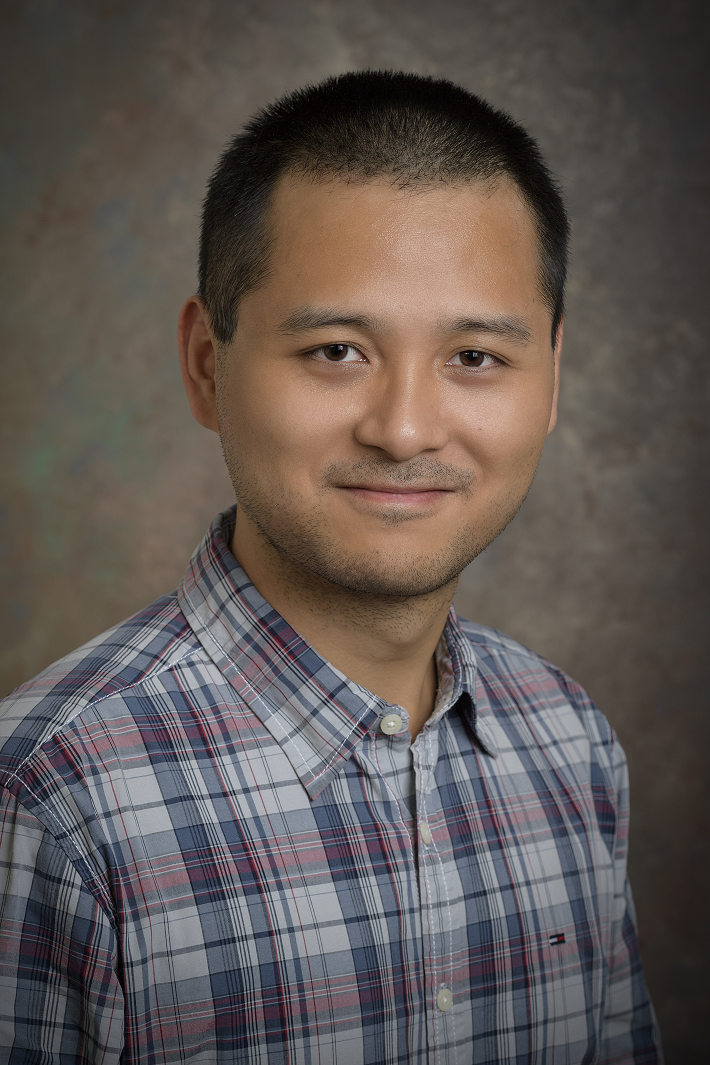}}]{Cencheng Shen received the BS degree in Quantitative Finance from National University of Singapore in 2010, and the PhD degree in Applied Mathematics and Statistics from Johns Hopkins University in 2015. He is an 
associate professor in the Department of Applied Economics and Statistics at University of Delaware. His research interests include graph inference, neural network, correlation and dependence.}
\end{IEEEbiography}
\begin{IEEEbiography}
[{\includegraphics[width=1in,height=1.25in,clip,keepaspectratio]{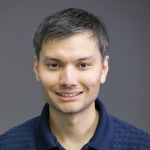}}]{Jonathan Larson is a Principal Data Architect at Microsoft Research working on Special Projects.  His applied research work focuses on petabyte-scale data infrastructure, data science applications, network analytics, and information visualization.  He has applied experience in organizational science, neuroscience, cyber-security, counter-human trafficking, fraud analytics, mobile device analytics, media management, retail analytics, and real estate. At Microsoft, Jonathan leads a research team of developers and data scientists focused on new approaches and applications for scalable network machine learning. }
\end{IEEEbiography}
\begin{IEEEbiography}
[{\includegraphics[width=1in,height=1.25in,clip,keepaspectratio]{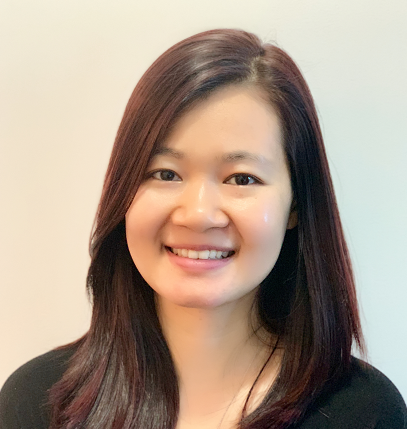}}]{Ha Trinh received a BS degree in Applied Computing from the University of Dundee in 2009 and a PhD degree in Computing from the same university in 2013. She is currently working as a data scientist at Microsoft Research. Her research interests lie at the intersection of Artificial Intelligence and Human-Computer Interaction.}
\end{IEEEbiography}
\begin{IEEEbiography}
[{\includegraphics[width=1in,height=1.25in,clip,keepaspectratio]{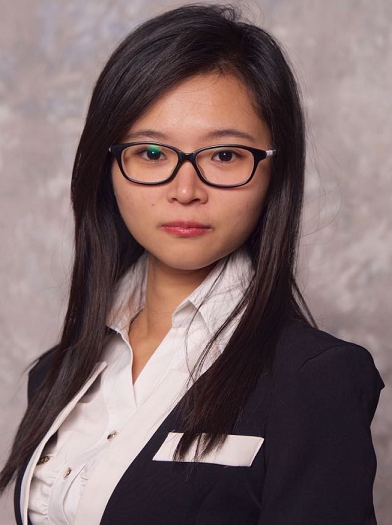}}]{Xihan Qin received the MS degree in Biotechnology from Georgetown University in 2015 and the MS degree in Bioinformatics from University of Delaware in 2021. She is currently a PhD candidate in Computer Science at University of Delaware. Her research interests include graph machine learning, bioinformatics, and computational biology. }
\end{IEEEbiography}
\begin{IEEEbiography}
[{\includegraphics[width=1in,height=1.25in,clip,keepaspectratio]{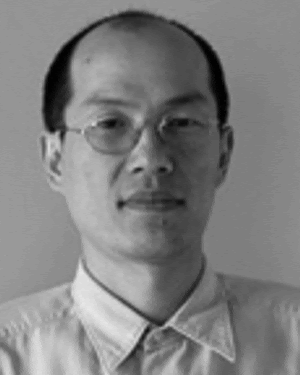}}]{Youngser Park received the B.E. degree in electrical engineering from Inha University in Seoul, Korea in 1985, the M.S. and Ph.D. degrees in computer science from The George Washington University in 1991 and 2011 respectively. From 1998 to 2000 he worked at the Johns Hopkins Medical Institutes as a senior research engineer. From 2003 until 2011 he worked as a senior research analyst, and has been an associate research scientist since 2011 then research scientist since 2019 in the Center for Imaging Science at the Johns Hopkins University. At Johns Hopkins, he holds joint appointments in the The Institute for Computational Medicine and the Human Language Technology Center of Excellence. His current research interests are clustering algorithms, pattern classification, and data mining for high-dimensional and graph data.}
\end{IEEEbiography}
\begin{IEEEbiography}
[{\includegraphics[width=1in,height=1.25in,clip,keepaspectratio]{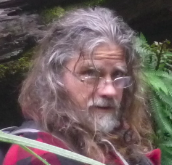}}]{Carey E. Priebe received the BS degree in mathematics from Purdue University in 1984, the MS degree in computer science from San Diego State University in 1988, and the PhD degree in information technology (computational statistics) from George Mason University in 1993. From 1985 to 1994 he worked as a mathematician and scientist in the US Navy research and development laboratory system. Since 1994 he has been a professor in the Department of Applied Mathematics and Statistics at Johns Hopkins University. His research interests include computational statistics, kernel and mixture estimates, statistical pattern recognition, model selection, and statistical inference for high-dimensional and graph data. He is a Senior Member of the IEEE, an Elected Member of the International Statistical Institute, a Fellow of the Institute of Mathematical Statistics, and a Fellow of the American Statistical Association.}
\end{IEEEbiography}

\clearpage
\setcounter{figure}{0}
\setcounter{theorem}{0}
\renewcommand{\thealgorithm}{C\arabic{algorithm}}
\renewcommand{\thefigure}{E\arabic{figure}}
\renewcommand{\thesubsection}{\thesection.\arabic{subsection}}
\renewcommand{\thesubsubsection}{\thesubsection.\arabic{subsubsection}}
\pagenumbering{arabic}
\renewcommand{\thepage}{\arabic{page}}

\bigskip
\begin{center}
{\large\bf APPENDIX}
\end{center}
\section{Proofs}
\label{sec:proofs}

The theoretical results assume conditional independent edge generation, i.e., each entry $A_{t}(i,j) | i,j$ is independently generated with finite second moments. This assumption covers most popular random graph models, including the stochastic block model \cite{SnijdersNowicki1997, KarrerNewman2011} and the degree-corrected variant \cite{ZhaoLevinaZhu2012}, as well as the random dot product models \cite{HollandEtAl1983, YoungScheinerman2007}. 

Another implicit assumption is $n_k =O(n)$, i.e., as vertex size increases to infinity, the number of vertices in each community also increase to infinity. This assumption is always satisfied by any probability model because if not, the prior probability of this community would be zero, making it non-existent.

\begin{theorem}
Assuming the conditional independent random graph model, the temporal encoder embedding converges to a conditional expectation. Specifically, for a vertex $i$ belonging to community $y$, we have that
\begin{align}
\label{eq2}
Z_{t}(i,k) \stackrel{n \rightarrow \infty}{\rightarrow} \frac{a_{t}(i,k)}{\|a_{t}(i,:)\|_{2}},
\end{align}
where $a_{t}(i,:) \in \mathbb{R}^{K}$ satisfies
\begin{align*}
a_{t}(i,k) = E(\mathbf{A}_{t}(i,j) | \mathbf{Y}_i=y, \mathbf{Y}_j=k).
\end{align*}
\end{theorem}
\begin{proof}
The convergence of Equation~\ref{eq2} can be proved by showing that the un-normalized embedding of vertex $i$ converges to $a_{t}(i,k)$ for each dimension $k$. Because once the convergence of the un-normalized embedding is established, the norm shall converge to $|a_{t}(i,:)|_{2}$.

To do so, we decompose step 3 for vertex $i$ and dimension $k$. Then the un-normalized embedding equals
\begin{align*}
\mathbf{A}_{t}(i,:)\mathbf{W}(:,k) =\frac{\sum_{j=1, j\neq i, Y_j=k}^{n} \mathbf{A}_{t}(i,j)}{n_k}.
\end{align*}
Note that since we already assumed $\mathbf{Y}_i=y$, all remaining equations are implicitly conditioned on $\mathbf{Y}_i=y$.

Its expectation satisfies 
\begin{align*}
E(\mathbf{A}_{t}(i,:)\mathbf{W}(:,k)) & = E( \frac{\sum_{j=1, j\neq i, Y_j=k}^{n} \mathbf{A}_{t}(i,j)}{n_k}) \\
& = \frac{\sum_{j=1, j\neq i, Y_j=k}^{n}E(\mathbf{A}_{t}(i,j)|Y_j=k )}{n_k} \\
& = E(\mathbf{A}_{t}(i,j)|Y_j=k )
\end{align*}
Note that the last equality holds if $Y_i \neq k$. In case of $Y_i=k$, the numerator on line 2 only has $n_k -1$ term, resulting in the need to adjust the last line to $\frac{n_k-1}{n_k} E(\mathbf{A}_{t}(i,j)|Y_j=k )$. However, as $n$ and $n_k$ approach infinity, the two cases become asymptotically equivalent, so it suffices to consider the former case for limiting $n$. 

Moreover, its variance satisfies
\begin{align*}
Var(\mathbf{A}_{t}(i,:)\mathbf{W}(:,k)) &= Var( \frac{\sum_{j=1, j\neq i, Y_j=k}^{n} \mathbf{A}_{t}(i,j)}{n_k})\\
&= \sum_{j=1, j\neq i, Y_j=k}^{n} \frac{Var(\mathbf{A}_{t}(i,j)|Y_j=k)}{n_k^2}\\
&\leq \frac{n_k M}{n_k^2}\\
& \stackrel{n \rightarrow \infty}{\rightarrow} 0,
\end{align*}
where $M$ denotes the maximum variance and is always bounded due to the finite second moments assumption. As the variance converges to $0$, by Chebychev inequality we have
\begin{align*}
\mathbf{A}_{t}(i,:)\mathbf{W}(:,k) \stackrel{n \rightarrow \infty}{\rightarrow} E(\mathbf{A}_{t}(i,j)|\mathbf{Y}_i=y, Y_j=k).
\end{align*}
Therefore, the un-normalized embedding of vertex $i$ converges to $a_{t}(i,k)$ for each dimension $k$, and the main theorem is proved.
\end{proof}

\begin{theorem}
Assuming non-zero conditional expectations, the dynamic statistic of vertex $i$ at time $t$ converges to $0$ asymptotically if and only if $\theta(a_{t}(i,:), a_{1}(i,:))=0$, and it converges to $1$ asymptotically if and only if $\theta(a_{t}(i,:), a_{1}(i,:))= \pi/2$, where $\theta(\cdot,\cdot)$ denotes the angle between two vectors.
\end{theorem}

\begin{proof}
The vertex dynamic statictic is defined by
\begin{align*}
\mbox{Dynamic}_{1,t}(i)=1-<\mathbf{\tilde{Z}}_{t}(i, \cdot), \mathbf{\tilde{Z}}_{1}(i, \cdot)>.
\end{align*}
By Theorem 1, we have
\begin{align*}
\mbox{Dynamic}_{1,t}(i) &\stackrel{n \rightarrow \infty}{\rightarrow} 1-<\frac{a_{t}(i,k)}{\|a_{t}(i,:)\|_{2}}, \frac{a_{1}(i,k)}{\|a_{1}(i,:)\|_{2}}> \\
& = 1- \cos\theta(a_{t}(i,:), a_{1}(i,:)).
\end{align*}
The results immediately follow.
\end{proof}

\section{Simulation Details}
The main paper's synthetic study is based on the degree-corrected stochastic block model. In the standard stochastic block model, each vertex $i$ is assigned a community label $Y_i \in \{1,\ldots, K\}$, which can be pre-determined or generated by a categorical distribution with prior probability $\{\pi_k \in (0,1) \mbox{ with } \sum_{k=1}^{K} \pi_k=1\}$. The edge probability between a vertex from community $k$ and a vertex from community $l$ is defined by a block probability matrix $\mathbf{B}=[\mathbf{B}(k,l)] \in [0,1]^{K \times K}$, and for any $i<j$ it holds that
\begin{align*}
\mathbf{A}(i,j) &\stackrel{i.i.d.}{\sim} \operatorname{Bernoulli}(\mathbf{B}(Y_i, Y_j)), \\
\mathbf{A}(i,i)&=0, \ \ \mathbf{A}(j,i) = \mathbf{A}(i,j).
\end{align*}
The degree-corrected stochastic block model is a generalization of SBM that accounts for the sparsity of real graphs. In addition to the parameters defined in SBM, each vertex $i$ is given a degree parameter $\theta_i$, and the adjacency matrix is generated by 
\begin{align*}
\mathbf{A}(i,j) \sim \operatorname{Bernoulli}(\theta_i \theta_j \mathbf{B}(Y_i, Y_j)).
\end{align*}
The degree parameters are usually constrained to ensure a valid probability. 

The standard stochastic block model generates dense graphs where all vertices within the same community have the same expected degrees. However, many real-world graphs exhibit sparsity and varying degrees among their vertices. Therefore, these degree parameters allow DC-SBM to better model the sparsity of each vertex and provide an accurate approximation for many real-world graphs, making it an ideal model for simulations. 

\subsection{Stable Network}
The synthetic study generated the graph at time $t=1$ using the degree-corrected stochastic block model. The model parameters were set as follows: $n=30000$, $K=20$, $Y_i=1,2,\ldots,20$ equally likely, and the block probability matrix satisfies: $\mathbf{B}(i,i)=0.5$ and $\mathbf{B}(i,j)=0.1$ for all $i=1,\ldots, 20$ and $j \neq i$. The degree parameter was independently generated by $Beta(1,4)$ for each vertex. 

Given the presence of an edge, the edge weights were randomized within the range $[1,...100]$. In other words, the adjacency matrix was generated as follows:
\begin{align*}
\mathbf{A}(i,j) \sim U_{ij} \cdot \operatorname{Bernoulli}(\theta_i \theta_j \mathbf{B}(Y_i, Y_j)),
\end{align*}
where $U_{ij}$ is equally likely to be any integer within $[1,100]$, and it is independently generated for different pairs of $(i,j)$.

For each subsequent time $t$, the graph at time $t$ was obtained by modifying the graph at time $t-1$ where $50\%$ of the edge weights were randomly changed. To introduce noise, a random number uniformly drawn from $[-20, +20]$ was added to the weight of each edge, and then the weights were enforced to be non-negative. In mathematical notation:
\begin{align*}
\mathbf{A}_{t}(i,j)=\max\{\mathbf{A}_{t}(i,j) + \epsilon_{ij} * Bernoulli(0.5),0\},
\end{align*}
where $\epsilon_{ij}$ is equally likely to be any integer within $[-20, 20]$, and it is independently generated for different pairs of $(i,j)$. This procedure resulted in a relatively stable time-series graph where the connectivity remained the same, but the edge weights gradually change over time.

\subsection{Outlier Vertices}
This simulation employed the same model settings as described above, with the exception of $n=1000$ vertices and $T=10$. At $t=10$, we introduced $10$ extreme outlier vertices, each with one or two incident edge weights randomly assigned within the range of $[500,1000]$, rendering them significantly different from all other edge weights.

For the unfolded spectral embedding, we experimented with all possible values of $d$ ranging from $1$ to $30$, and the results were consistently similar to those obtained with $d=10$.

\subsection{Pattern Shift}
In this simulation, we considered a binary graph with $n=30000$ vertices, $K=3$, equally likely labels $Y_i=1,2,3$, and independently generated degree parameters following a $Beta(1,4)$ distribution for each vertex. 

At $t=1$, we generated the initial graph using the block probability matrix $\mathbf{B}_1 \in \mathbb{R}^{3 \times 3}$ with $\mathbf{B}_1(i,i)=0.9$ and $\mathbf{B}_1(i,j)=0.1$ for all $i=1,\ldots, 3$ and $j \neq i$.

At $t=2$, the block probability is changed for community $3$. We first set $\mathbf{B}_2=\mathbf{B}_1$ and then modified $\mathbf{B}_{2}(:,3)=0.3$, meaning that vertices in community $3$ had a $30\%$ chance of having an edge with any other vertex, up to the degree parameter.

At $t=3$, the label vector is changed: each vertex in community $3$ was randomly reassigned to a new community $4$ with a $50\%$ probability. The block probability matrix $\mathbf{B}_3 \in \mathbb{R}^{4 \times 4}$ was set as follows: $\mathbf{B}_3(i,i)=0.9$ and $\mathbf{B}_3(i,j)=0.1$ for all $i=1,\ldots, 4$ and $j \neq i$. Then, we set $\mathbf{B}_{3}(:,3)=0.3$. Essentially, the edge probability for communities $1$, $2$, and $3$ remained the same as at $t=2$, while community $4$ reverted to the edge probability of community $3$ at $t=1$.

At $t=4$, two communities are merged: all of community $3$ became community $1$. Vertices in community $3$ follows the same edge probability as vertices in community $1$.
 
The temporal encoder embedding was applied to all four graphs, with the initial label vector used as input, generating an $n \times 3 \times 4$ embedding. Note that the modified label vector was not used in the embedding process.
\end{document}